\documentclass[conference]{IEEEtran}
\usepackage[utf8]{inputenc}
\usepackage[style=ieee, maxbibnames=99, minbibnames=99]{biblatex}
\usepackage{amsmath}
\usepackage{amssymb}
\usepackage{systeme}
\usepackage{graphicx}
\usepackage{caption}
\usepackage{cleveref}
\usepackage{bm}
\usepackage[linesnumbered,ruled,vlined]{algorithm2e}
\usepackage{amsthm}
\usepackage{spalign}
\usepackage{amsfonts}
\usepackage{enumitem}
\usepackage{multirow}
\usepackage{array}

\graphicspath{ {./images/} }

\let\oldnl\nl
\newcommand{\nonl}{\renewcommand{\nl}{\let\nl\oldnl}}

\newtheorem{theorem}{Theorem}

\newcommand{\hvw}{\Hat{v}_w}

\newcommand{\by}{\mathbf{y}}

\newcommand{\bx}{\mathbf{x}}
\newcommand{\bw}{\mathbf{w}}
\newcommand{\bh}{\mathbf{h}}
\newcommand{\bq}{\mathbf{q}}

\newcommand{\br}{\mathbf{r}}
\newcommand{\bp}{\mathbf{p}}
\newcommand{\bz}{\mathbf{z}}

\newcommand{\bg}{\mathbf{g}}

\newcommand{\bzero}{\mathbf{0}}
\newcommand{\bmu}{\bm{\mu}}
\newcommand{\bLambda}{\bm{\Lambda}}
\newcommand{\bA}{\mathbf{A}}
\newcommand{\bI}{\mathbf{I}}
\newcommand{\bU}{\mathbf{U}}
\newcommand{\bV}{\mathbf{V}}
\newcommand{\bS}{\mathbf{S}}
\newcommand{\bW}{\mathbf{W}}
\newcommand{\bH}{\mathbf{H}}
\newcommand{\bs}{\mathbf{s}}

\newcommand{\bZ}{\mathbf{Z}}

\newcommand{\bP}{\mathbf{P}}

\newcommand{\bD}{\mathbf{D}}

\newcommand{\bgamma}{\bm{\gamma}}
\newcommand{\bPsi}{\bm{\Psi}}

\newcommand{\bsigma}{\bm{\sigma}}

\newcommand{\bone}{\bm{1}}

\newcommand{\normDensity}{\mathcal{N}}
\newcommand{\expectation}{\mathbb{E}}

\newcommand{\lambdaDag}{\lambda^{\dagger}}

\newcommand{\limN}{\lim \limits_{N \rightarrow \infty}}
\newcommand{\as}{\overset{a.s.}{=}}
\newcommand{\inv}[1]{\frac{1}{#1}}

\SetKwInput{KwInput}{Input}               
\SetKwInput{KwOutput}{Output}
\SetKwInOut{KwInitialization}{Initialization}
\SetKwBlock{KwBlockA}{Block A}{end}
\SetKwBlock{KwBlockB}{Block B}{end}

\SetKwBlock{KwBlockCGRoutine}{Standard CG routine}{end}
\SetKwBlock{KwBlockCGDivergence}{CG correction}{end}
\SetKwBlock{KwBlockCGMSE}{SE of Block A}{end}

\title{Warm-Starting in Message Passing algorithms}

\IEEEoverridecommandlockouts

\author{\IEEEauthorblockN{Nikolajs Skuratovs, Mike E. Davies}
\IEEEauthorblockA{Institute for Digital Communications, School of Engineering, University of Edinburgh, Edinburgh, EH9 3FG, U.K.}
\thanks{This work was supported by the ERC project C-SENSE (ERC-ADG-2015-694888). MD is also supported by a Royal Society Wolfson Research Merit Award.}}

\begin{document}

\maketitle

\begin{abstract}
    Vector Approximate Message Passing (VAMP) provides the means of solving a linear inverse problem in a Bayes-optimal way assuming the measurement operator is sufficiently random. However, VAMP requires implementing the linear minimum mean squared error (LMMSE) estimator at every iteration, which makes the algorithm intractable for large-scale problems. In this work, we present a class of warm-started (WS) methods that provides a scalable approximation of LMMSE within VAMP. We show that a Message Passing (MP) algorithm equipped with a method from this class can converge to the fixed point of VAMP while having a per-iteration computational complexity proportional to that of AMP. Additionally, we provide the Onsager correction and a multi-dimensional State Evolution for MP utilizing one of the WS methods. Lastly, we show that the approximation approach used in the recently proposed Memory AMP (MAMP) algorithm is a special case of the developed class of WS methods. 
\end{abstract}

\section{Introduction}

Consider recovering a random signal $\bx \in \mathbb{R}^N$ from 

\noindent
\begin{equation}
    \by = \bA \bx + \bw \label{eq:y_measurements}
\end{equation}

\noindent
where $\by \in \mathbb{R}^M$, $\bw \in \mathbb{R}^M$  is a zero-mean i.i.d. Gaussian noise vector $\bw \sim N(0,v_w \bI_M)$ and $\bA \in \mathbb{R}^{M \times N}$ is a measurement matrix that is assumed to be available. We consider the limiting case where $N = \delta M \rightarrow \infty$ with $\frac{M}{N} \rightarrow \delta \in (0,1)$.

In this work we approach the problem of recovering $\bx$ from \eqref{eq:y_measurements} by referring to the family of Approximate Message Passing (AMP) algorithms originally proposed in \cite{AMP}. When $\bA$ is drawn from the class of sub-Gaussian matrices, AMP demonstrates fast convergence rate, stability and existence of a 1D dynamics called \textit{State Evolution (SE)} that defines the evolution of the intrinsic uncertainty in the algorithm \cite{SE_AMP}. In particular, the SE was used to show the Bayes-optimality of the algorithm \cite{AMP_SE_non_separable, AMP_Bayes_ioptimality_1, AMP_Bayes_ioptimality_2} in the large system limit (LSL).

Yet, when the measurement operator $\bA$ is ill-conditioned or non-centered, AMP was evidenced to be unstable and provide poor reconstruction performance \cite{VAMP, AMP_convergence_general_A, AMP_unstable_general_A}. To expand the class of allowed measurement matrices, it was suggested to reformulate AMP as a symmetric procedure involving two steps that follow the orthogonality principle \cite{OAMP}: a denoising step and a linear step. This procedure was proposed in \cite{EP_Keigo}, \cite{OAMP} and \cite{VAMP}, where the latter two works called it \textit{Orthogonal AMP (OAMP)} and \textit{Vector Approximate Message Passing (VAMP)} respectively. In \cite{EP_Keigo, VAMP}, it was shown that OAMP/VAMP\footnote{While the LMMSE-OAMP and VAMP are almost identical, in the following we will stick to the name "VAMP" when we refer to these type of algorithms.} can operate with an ill-conditioned measurement operator and preserves most of the main advantages of AMP, when $\bA$ is right orthogonally invariant (ROI), which is a much wider class of random matrices. In this setup, it was shown that in the large system limit, the dynamics of VAMP can be defined through a 1D SE and the algorithm achieves Bayes optimal reconstruction conditioned on the validity of the replica prediction for ROI matrices \cite{VAMP}. 

Despite the advantages of VAMP, each of its iteration implements an LMMSE estimator that requires inverting an $M$ by $M$ matrix, which is computationally infeasible for large-scale inverse problems. In \cite{VAMP} it was proposed to leverage the SVD decomposition of $\bA$ to reduce the computational cost, but this would require storing large singular vector matrices, which becomes infeasible from the memory point of view. To resolve this computational/memory bottleneck, several Message Passing (MP) techniques were proposed, including \textit{Conjugate Gradient VAMP (CG-VAMP)} \cite{CG_EP, D-VAMP_poster, OurPaper}, \textit{Convolutional AMP (CAMP)} \cite{CAMP_short, CAMP_full}, \textit{Memory AMP (MAMP)} \cite{MAMP_full, MAMP_short} and others. CAMP and MAMP are based on the general OAMP framework \cite{UnifiedSE}, where the optimal LMMSE estimator is approximated by a long-memory (LM) matched filter (LM-MF), which utilizes the information from the previous iterations. In \cite{CAMP_full, MAMP_full} it was shown that the dynamics of these algorithms can be described by an LM version of SE, which was used to prove the Bayes-optimality of the algorithms under the same assumptions as for VAMP. Importantly, these algorithms do not require storing any large dimensional matrices and have the computational complexity similar to that of AMP. While CAMP demonstrates consistent performance for well-to-moderate conditioned operators $\bA$, MAMP was evidenced to be stable for highly ill-conditioned operators \cite{MAMP_full}. However, LM-MF has a fixed quality of approximation and, when $\bA$ is ill-conditioned, the number of iterations required for MAMP to converge might be hundreds and thousands \cite{MAMP_full}, which can have a substantial computational cost especially if the denoiser used in the algorithm is relatively expensive, as with Non-Local Mean \cite{NLM} , BM3D \cite{BM3D} etc.

In this work we upscale VAMP by designing a class of warm-started (WS) first-order methods approximating LMMSE with a flexible number of inner-loop iterations leading to efficient operation of the algorithm. As special cases, this class includes WS Gradient Descent (WS-GD) and WS Conjugate Gradient (WS-CG). We show that when an MP algorithm is equipped with a method from this class, the algorithm converges to the same fixed point as VAMP. We provide a closed-form solution for the Onsager correction and the LM SE for this class of approximators and present a robust method for implementing them. Additionally, we show that the approximation method used in MAMP corresponds to a particular form of WS-GD with one inner-loop iteration. Lastly, we numerically compare the performance of the proposed methods against VAMP, MAMP and regular CG-VAMP.

\section{Long Memory VAMP}

In this section we focus on a particular instance of the general LM MP framework \cite{UnifiedSE} that is used as the basis for MAMP \cite{MAMP_full} and for the proposed algorithm. Let $\br_t \in \mathbb{R}^N$ and $\bs_t \in \mathbb{R}^N$ be iteratively updated as\footnote{In MAMP there are additional scalar terms involved that ensure stability and the convergence properties, but to simplify the exposition we omit them here. However, we utilize the optimal MAMP algorithm \cite{MAMP_full} when it comes to the numerical analysis.}

\noindent
\begin{equation}
    \br_t = \Big(\sum \limits_{\tau=0}^t \gamma_t^{\tau} \Big)^{-1} \Big( \sum_{\tau=0}^t \gamma_t^{\tau} \bs_{\tau} + \bA^T \bmu_t(\bS_t) \Big) \label{eq:r_t}
\end{equation}

\noindent
\begin{equation}
    \bs_{t+1} = \inv{1 - \alpha_t} \Big( \bg( \br_t ) - \alpha_t \br_t \Big) \label{eq:s_t}
\end{equation}


\noindent
where  $\bS_t = \big\{ \bs_t, \bs_{t-1}, ..., \bs_0 \big\}$ and $\bmu_t$ is a long-memory approximation of the LMMSE estimator 

\noindent
\begin{equation}
    \bmu_t^{LMMSE} = \Big( \rho_t \bI + \bA \bA^T \Big)^{-1} \big( \by - \bA \bs_t \big) = \bW_t^{-1} \bz_t  \label{eq:LMMSE}
\end{equation}

\noindent
In \eqref{eq:r_t}-\eqref{eq:s_t}, the scalars $\alpha_t$ and $\gamma_t^{\tau}$ correspond to the following divergences

\noindent
\begin{equation}
    \gamma_t^{\tau} = \inv{N} \nabla_{\bs_{\tau}} \cdot \bA^T \bmu_t(\bS_t) \quad \quad \quad \alpha_t = \inv{N} \nabla_{\br_t} \cdot \bg(\br_t) \label{eq:gamma_alpha_general}
\end{equation}

\noindent
Here and in \eqref{eq:s_t}, the function $\bg$ performs denoising of $\bx$ assuming an instrinsic Gaussian channel

\noindent
\begin{equation}
    \br_t = \bx + \bh_t \quad \quad \quad (\bh_t)_n \sim \normDensity(0, v_{h_t}) \label{eq:h_t}
\end{equation}

\noindent
Lastly, in \eqref{eq:LMMSE}  we have the scalar $\rho_t = \frac{v_w}{v_{q_t}}$ with $v_{q_t}$ modeling the variance $\inv{N} \expectation\big[ ||\bq_t||^2 \big]$ of the error vector $\bq_t = \bs_t - \bx$. 

With these definitions, one can establish rigorous dynamics of \eqref{eq:r_t}-\eqref{eq:s_t} in the LSL and under the following main assumptions:

\textbf{Assumption 1}: The matrix $\bA$ is right-orthogonally invariant (ROI) such that in the SVD of $\bA = \bU \bS \bV^T$, the matrix $\bV$ is independent of other random terms and is Haar distributed \cite{rand_mat_methods_book}. Additionally, let the limiting spectral distribution of $\bLambda = \bS \bS^T$ almost surely converge to a compactly supported function and be normalized such that $\inv{N} Tr \{\bS \bS^T\} = 1$.

\textbf{Assumption 2}: The sequence of functions $\bg: \mathbb{R}^N \mapsto \mathbb{R}^N$ indexed by $N$ are Lipschitz continuous with a Lipschitz constant $L_N < \infty$ as $N \rightarrow \infty$ \cite{NS-VAMP}, \cite{AMP_SE_non_separable}. Additionally, we assumes the sequences of $\bx$ and $\bg(\cdot)$ are \textit{convergent under Gaussian noise} \cite{NS-VAMP}.

Under these assumptions, one can ensure that the intrinsic model \eqref{eq:h_t} holds almost surely and that the error vectors $\bq_t$ and $\bh_t$ are asymptotically orthogonal \cite{UnifiedSE}

\noindent
\begin{equation}
    \limN \inv{N} \bq_t^T \bh_t \as 0 \label{eq:h_q_independence}
\end{equation}

\section{Long-memory LMMSE approximation} \label{sec:LM_LMMSE_approximators}

As seen from \eqref{eq:LMMSE}, the direct implementation of LMMSE requires inverting an $M$ by $M$ matrix. To overcome this issue, first, we reformulate $\bmu_t^{LMMSE}$ as the solution to the following system of linear equations (SLE)

\noindent
\begin{equation}
    \bW_t \bmu_t = \bz_t \label{eq:SLE}
\end{equation}

\noindent
Next, we define a quadratic cost function \cite{CG_properties, Painless_CG, Optimization_polyk}

\noindent
\begin{equation}
    f_t(\bmu) = \inv{2} \bmu^T \bW_t \bmu - \bz_t^T \bmu \label{eq:cost_function} 
\end{equation}

\noindent
By setting $\nabla f_t(\bmu) = \bW_t \bmu - \bz_t = 0$, one can verify that $f_t(\bmu)$ is minimized by the LMMSE solution \eqref{eq:LMMSE}. Then, by constructing an iterative method for minimizing \eqref{eq:cost_function}, one can potentially recover $\bmu_t^{LMMSE}$. This idea was previously considered in the similar works including \cite{CG_EP,D-VAMP_poster,OurPaper} but the approximation methods there did not utilize the \textit{warm-start (WS)} approach. In this work, we consider applying WS to the following iterative scheme

\noindent
\begin{equation}
    \bmu_t^{i+1} = \bmu_t^i + a_t^i \bp_t^i \quad \quad \bp_t^{i+1} = -\nabla f_t(\bmu_t^{i+1}) + b_t^i \bp_t^i \label{eq:general_iterative_scheme}
\end{equation}

\noindent
where $\bmu_t^i$ is the approximation of \eqref{eq:LMMSE} after $i$ iterations, $\bp_t^i$ is the new search direction and $a_t^i$ is the step size along $\bp_t^i$. One special case of \eqref{eq:general_iterative_scheme} is the \textit{Gradient Descent (GD)} algorithm \cite{Painless_CG, Optimization_polyk} for which we set $b_t^i = 0$ for all $i$ and choose a fixed step size $a_t^i = a_t$ to be

\noindent
\begin{equation}
    a_t = \frac{2}{L_{max}(\bW_t) + L_{min}(\bW_t)} = \inv{\rho_t + \lambda^{\dagger}} \label{eq:optimal_GD_step_size}
\end{equation}

\noindent
where $L_{min}(\bW_t)$ and $L_{max}(\bW_t)$ are the minimum and maximum eigenvalues of $\bW_t$ and $\lambdaDag = \frac{L_{max}(\bS \bS^T) + L_{min}(\bS \bS^T)}{2}$. Similarly, by defining

\noindent
\begin{equation}
    a_t^i = \frac{||\bz_t - \bW_t \bmu_t^i||^2}{ (\bp_t^i)^T \bW_t \bp_t^i} \quad \quad b_t^i = \frac{|| \bz_t - \bW_t \bmu_t^{i+1} ||^2}{|| \bz_t - \bW_t \bmu_t^i ||^2} \label{eq:CG_scalars}
\end{equation}

\noindent
one obtains the \textit{Conjugate Gradient (CG)} method. One can show \cite{Painless_CG, Optimization_polyk} that CG is the optimal first-order method for minimizing \eqref{eq:cost_function} and has a substantially faster convergence rate compared to GD when the matrix $\bW_t$ is ill-conditioned. 

Although both of the methods are guaranteed to converge to the exact solution \eqref{eq:LMMSE} \cite{Optimization_polyk}, the number of inner-loop iterations required for that is $i=M$ for CG and $i \rightarrow \infty$ for GD, which is as expensive as directly computing \eqref{eq:LMMSE}. To keep the number of inner-loop iterations of \eqref{eq:general_iterative_scheme} of order $O(1)$ while preserving the ability to recover the exact LMMSE solution, we suggest combining a small number of inner loop iterations with the following warm-start at each outer loop iteration

\noindent
\begin{equation}
    \bmu_t^0 = \bmu_{t-1}^i  \quad \quad \bp_t^{0} = -\nabla f_t(\bmu_{t-1}^i) +  b_{t-1}^{i-1} \bp_{t-1}^{i-1} \label{eq:general_initialization}
\end{equation}

\noindent
and $\bmu_0^0 = \bzero$, $\bp_0^0 = \bz_t$ for $t=0$. The following theorem establishes the asymptotic property of an MP algorithm utilizing \eqref{eq:general_iterative_scheme} with the initialization \eqref{eq:general_initialization}.

\begin{theorem} \label{theorem:WS_iterative_approximation_convergence}

    Consider the MP algorithm \eqref{eq:r_t}-\eqref{eq:s_t} with $\bmu_t = \bmu_t^i$ being the output of the iterations \eqref{eq:general_iterative_scheme} with $i>0$ and the initialization \eqref{eq:general_initialization}. Assume that there is such $i^{\ast}$ so that $\bmu_t^{i^{\ast}}$ minimizes $f_t(\cdot)$ and that the MP algorithm \eqref{eq:r_t}-\eqref{eq:s_t} has a unique fixed point achieved at iteration $t^{\ast}$. Then, under Assumptions 1-2, the vector $\bmu_{t^{\ast}}$ converges to $\bmu_{t^{\ast}}^{LMMSE}$.
    
\end{theorem}

\begin{proof}
    Let the MP algorithm \eqref{eq:r_t}-\eqref{eq:s_t} converge to a unique fixed point at iteration $t^{\ast}$. Consider the iteration $t^{\ast}+1$ and use \eqref{eq:general_initialization} to obtain
    
    \noindent
    \begin{equation*}
        \bmu_{(t^{\ast}+1)}^0 = \bmu_{t^{\ast}}^i  \quad \quad \bp_{(t^{\ast}+1)}^0 = -\nabla f_{t^{\ast}+1}(\bmu_{t^{\ast}}^{i}) + b_{t^{\ast}}^{i-1} \bp_{t^{\ast}}^{i-1}
    \end{equation*}
    
     \noindent
    Because $\bW_{t^{\ast}+1} = \bW_{t^{\ast}}$ and $\bz_{t^{\ast}+1} = \bz_{t^{\ast}}$, we have that $\nabla f_{t^{\ast}+1}(\cdot) = \nabla f_{t^{\ast}}(\cdot)$, $a_{t^{\ast}+1}^j = a_{t^{\ast}}^j$ and $b_{t^{\ast}+1}^j = b_{t^{\ast}}^j$. Therefore, by running the iterative scheme \eqref{eq:general_iterative_scheme} for $i$ iterations, we produce $\bmu_{(t^{\ast}+1)}^i = \bmu_{(t^{\ast})}^{2i}$ and $\bp_{(t^{\ast}+1)}^i = \bp_{(t^{\ast})}^{2i}$. Then, after at most $\frac{i^{\ast}}{i}$ outer-loop iterations, we obtain $\bmu_{t^{\ast}}^{i^{\ast}}$, which, by the assumption, corresponds to $\bmu_{t^{\ast}}^{LMMSE}$. \end{proof}

\subsection*{MAMP}

The recently proposed MAMP algorithm \cite{MAMP_full} approximates \eqref{eq:LMMSE} with a LM approximation method based on Taylor expansion of the matrix $\bW_t$. Using the ideas above, we can show that this approximation method corresponds to a particular WS-GD scheme with $i=1$ inner-loop iteration. In particular, in MAMP \eqref{eq:LMMSE} is approximated as \cite{MAMP_full}

\noindent
\begin{equation}
    \bmu_t = \big( \bI - a_t \bW_t \big) \bmu_{t-1} + \Bar{a}_t \bz_t \label{eq:MAMP_mu}
\end{equation}

\noindent
where $a_t$ is the same as in \eqref{eq:optimal_GD_step_size}, while $\Bar{a}_t$ is optimized to minimize the variance $v_{h_t}$ from \eqref{eq:h_t}. To relate \eqref{eq:MAMP_mu} to \eqref{eq:general_iterative_scheme}, we let $\Bar{a}_t = c_t a_t$, where $c_t$ is some scalar, and add the corresponding superscript $i$. Then, after rearranging the terms in \eqref{eq:MAMP_mu}, one can obtain

\noindent
\begin{equation}
    \bmu_t^1 = \bmu_{t-1}^0 + a_t^0 \bp_t^0 \quad \quad \bp_t^0 = - \big( \bW_t \bmu_{t-1}^0 - c_t \bz_t \big) \label{eq:MAMP_mu_rearranged}
\end{equation}

\noindent
By comparing \eqref{eq:MAMP_mu_rearranged} to \eqref{eq:general_iterative_scheme}, we conclude that this update corresponds to WS-GD with $i=1$ approximating the cost

\noindent
\begin{equation}
    f_t(\bmu) = \inv{2} \bmu^T \bW_t \bmu - c_t \bz_t^T \bmu \label{eq:MAMP_cost}
\end{equation}

\noindent
that has a minimum at $c_t \bW_t^{-1} \bz_t$. By setting $\bmu_t^0 = \bmu_{t-1}^1$ for $t>0$ and $\bmu_0^0 = \bzero$, $b_t^0 = 0$ and $a_t^0 = \inv{\rho_t + \lambdaDag}$, and applying Theorem \ref{theorem:WS_iterative_approximation_convergence} to \eqref{eq:MAMP_mu_rearranged} and \eqref{eq:MAMP_cost}, we confirm that $\bmu_t$ converges to $\bmu_{t^{\ast}} = c_{t^{\ast}} \bW_{t^{\ast}}^{-1} \bz_{t^{\ast}} = c_{t^{\ast}} \bmu_{t^{\ast}}^{LMMSE}$. However, due to the normalization scalars $\gamma_t^{\tau}$ from \eqref{eq:gamma_alpha_general}, we can show that the update $\br_t$ is invariant to scalar scaling of $\bmu_t$. Let $\bmu_t = c_t \bmu_t^i$ be the approximation of $\bmu_t^{LMMSE}$ used in $\br_t$ and $\gamma_t^{\tau, i} = \inv{N} \nabla_{\bs_{\tau}} \cdot \bA^T \bmu_t^i$. Then $\gamma_t^{\tau}$ is equivalent to

\noindent
\begin{equation*}
    \gamma_t^{\tau} = \inv{N} \nabla_{\bs_{\tau}} \cdot \bA^T \bmu_t(\bS_t) = \inv{N} \nabla_{\bs_{\tau}} \cdot \bA^T c_t \bmu_t^i = c_t \gamma_t^{\tau, i} 
\end{equation*}

\noindent
Substituting this into \eqref{eq:r_t} leads to

\noindent
\begin{equation*}
    \br_t = \frac{ \sum_{\tau=0}^t \gamma_t^{\tau} \bs_{\tau} + \bA^T \bmu_t }{\sum \limits_{\tau=0}^t \gamma_t^{\tau}} = \frac{ \sum_{\tau=0}^t \gamma_t^{\tau, i} \bs_{\tau} + \bA^T \bmu_t^i }{\sum \limits_{\tau=0}^t \gamma_t^{\tau, i}}
\end{equation*}

\noindent
which is the same as if $\bmu_t = \bmu_t^i$ without the scaling by $c_t$. Thus, the fixed point $\br_{t^{\ast}}$ of MAMP is the same with either \eqref{eq:MAMP_mu_rearranged} or \eqref{eq:general_iterative_scheme}-\eqref{eq:optimal_GD_step_size}.

\section{Asymptotic behaviour of WS methods within Message Passing}

While Theorem \ref{theorem:WS_iterative_approximation_convergence} suggests that a LM MP algorithm utilizing a WS method \eqref{eq:general_iterative_scheme} and \eqref{eq:general_initialization} has a potential to converge to the fixed point of VAMP, it is still required to identify the closed-form solutions for the correction scalars $\gamma_t^{\tau, i}$ and the evolution model of $v_{h_t}^i$. Next, we provide these closed-form solutions for any WS method that fits into the form of \eqref{eq:general_iterative_scheme} such as WS-CG, WS-GD (including the one from MAMP) and potentially other methods. To do this, define the following base sequences

\noindent
\begin{gather}
    f^{t,i}[j] = f^{t,i-1}[j] + a_t^i g^{t,i-1}[j] \label{eq:f_t_z} \\
    g^{t,i}[j] = - \rho_t f^{t,i}[j] - f^{t,i}[j-1] + b_t^i g^{t,i-1}[j] \label{eq:g_t_z}
\end{gather}

\noindent
where $a_t^i$ and $b_t^i$ depend on the chosen algorithm and $\rho_t$ is as in \eqref{eq:LMMSE}. The following table defines the three pairs of additional sequences for $i=1,2,...$, where $\delta_{j,0}$ is the Kronecker delta function.

\noindent
\begin{minipage}{\linewidth}
\centering
\captionof{table}{Supplementary sequences}
\begin{tabular}{|m{8.5em}| c| m{6.6em}||} 
    \hline
    Sequence definition & Initialization & Boundaries \\
    \hline\hline
    $f_{\bz}^{t,i}[j]=f^{t,i}[j]$ & $f_{\bz}^{t,0}[0]=0$ & $j \in \{0,...,i \}$ \\ 
    \hline
    $g_{\bz}^{t,i}[j]=g^{t,i}[j] + \delta_{j,0}$ & $g_{\bz}^{t,0}[0]=1$ & $j \in \{0,...,i \}$ \\
    \hline \hline
    $f_{\bmu}^{t,i}[j]=f^{t,i}[j]$ & $f_{\bmu}^{t,0}[0]=1$ & $j \in \{0,...,i \}$ \\
    \hline
    $g_{\bmu}^{t,i}[j] = g^{t,i}[j]$ & $\begin{matrix} g_{\bmu}^{t,0}[0]=-\rho_t \\ g_{\bmu}^{t,0}[1]=-1 \end{matrix}$ & $j \in \{0,...,i+1 \}$ \\
    \hline \hline
    $f_{\bp}^{t,i}[j]=f^{t,i}[j]$ & $f_{\bp}^{t,0}[0]=0$ & $j \in \{0,...,i \}$ \\ 
    \hline 
    $g_{\bp}^{t,i}[j] = g^{t,i}[j]$ & $g_{\bp}^{t,0}[0]=b_{t-1}^{i-1}$ & $j \in \{0,...,i \}$ \\ 
    \hline
\end{tabular}
\end{minipage}
    
\smallskip    
    
With these definitions, we are ready to set up the main sequences. For $0 \leq \tau < t$, let $r_{\tau}^{t,d}[k]$ and $u_{\tau}^{t,d}[k]$ be defined as 

\noindent
\begin{gather*}
     r_{\tau}^{t,d}[k] = \sum_{j=0}^{d}  f_{\bmu}^{t,d}[j] r_{\tau}^{t-1,d}[k-j] + f_{\bp}^{t,d}[j] u_{\tau}^{t-1,i-1}[k-j] \label{eq:r_t_lemma} \\
     u_{\tau}^{t,d}[k] = \sum_{j=0}^{d} g_{\bmu}^{t,d}[j] r_{\tau}^{t-1,d}[k-j] + g_{\bp}^{t,d}[j] u_{\tau}^{t-1,i-1}[k-j] \label{eq:u_t_lemma}
\end{gather*}

\noindent
with $r_{\tau}^{t,i}[k]=0$ and $u_{\tau}^{t,i}[k]=0$ for $k \not\in \{0,...,(t-\tau)i\}$. For $\tau=t$ we set $r_t^{t,i}[k] = f_{\bz}^{t,i}[k]$ and $u_t^{t,i}[k] = g_{\bz}^{t,i}[k]$. The following theorem establishes an equivalent model of $\bmu_t^i$ from \eqref{eq:general_iterative_scheme} using the above definitions.

\begin{theorem} \label{theorem:bz_tau_to_mu}
    
    \textit{\cite{OurFullPaper}}: Let $\bmu_t^i$ be the output of iterations \eqref{eq:general_iterative_scheme} with the initializations \eqref{eq:general_initialization}. Then
    
    \noindent
    \begin{equation}
        \bmu_t^i = \sum_{\tau = 0}^{t}  \sum_{k = 0}^{(t - \tau) i} r_{\tau}^{t,i}[k] \big( \bA \bA^T \big)^k \bz_{\tau} \label{eq:bz_tau_to_bmu}\\
    \end{equation}
    
\end{theorem}

This result can be used to work out the evolution of $v_{h_t}^i$, which in the limit corresponds to

\noindent
\begin{equation}
    v_{h_t}^i \as \lim_{N \rightarrow \infty} \frac{||\br_t^i - \bx||^2}{N} \overset{(a)}{=} \lim_{N \rightarrow \infty} \frac{\Big|\Big| \sum \limits_{\tau=0}^t \gamma_t^{\tau,i} \bq_{\tau} + \bA^T \bmu_t^i \Big| \Big|^2}{C_t^2 N} \nonumber\\
\end{equation}

\noindent
\begin{equation}
    = \lim_{N \rightarrow \infty} \frac{\big| \big| \sum \limits_{\tau=0}^t \gamma_t^{\tau,i} \bq_{\tau} \big| \big|^2 +2 \sum \limits_{\tau=0}^t \gamma_t^{\tau,i} \bq_{\tau}^T \bA^T \bmu_t^i +\big| \big| \bA^T \bmu_t^i \big| \big|^2}{C_t^2 N} \nonumber\\
\end{equation}

\noindent
\begin{equation}
    \as \inv{C_t^2} \Big( \lim_{N \rightarrow \infty} \frac{\big| \big|\bA^T \bmu_t^i \big| \big|^2}{N} - \sum_{\tau, \tau' = 0}^t \psi_{\tau,\tau'} \gamma_t^{\tau,i} \gamma_t^{\tau',i} \Big) \label{eq:v_h_WS}
\end{equation}


\noindent
where we defined $C_t = \sum \limits_{\tau=0}^t \gamma_t^{\tau, i}$ and $\psi_{\tau, \tau'} = \lim \limits_{N \rightarrow \infty} \inv{N} \bq_{\tau}^T \bq_{\tau'}$. Here, the step (a) comes from \eqref{eq:r_t}, while the last step in \eqref{eq:v_h_WS} is due to the following asymptotic identity \cite{UnifiedSE}

\noindent
\begin{equation}
    \limN \inv{N} \bq_{\tau}^T \bA^T \bmu_t^i \as - \sum_{\tau = 0}^t \psi_{t,\tau} \gamma_t^{\tau,i} \label{eq:q_A_mu_WS_CG}
\end{equation}

\noindent
Then, by substituting the result from Theorem \ref{theorem:bz_tau_to_mu} into \eqref{eq:v_h_WS}, one can obtain the closed-form solution for $v_{h_t}^i$. Similarly, one can plug the identity for $\bmu_t^i$ into \eqref{eq:gamma_alpha_general} to obtain the solution for $\gamma_t^{\tau,i}$. This idea is presented in the following theorem.

\begin{theorem} \label{theorem:WS_CG_correction_and_SE}
    
    \textit{\cite{OurFullPaper}}: Let $\bmu_t^i$ be the output of iterations \eqref{eq:general_iterative_scheme} with the initializations \eqref{eq:general_initialization}. Then, under Assumptions 1-2, the scalar $\gamma_t^{\tau,i}$ from \eqref{eq:gamma_alpha_general} is equivalent to
     
    \noindent
    \begin{equation}
        \limN \gamma_t^{\tau,i} \overset{a.s.}{=} \sum_{k = 0}^{(t - \tau) i} r_{\tau}^{t,i}[k] \chi_{k+1} \label{eq:gamma_WS_CG}
    \end{equation}
    
    \noindent
    where 
    
    \noindent
    \begin{equation}
        \chi_j = \limN \inv{N} Tr \Big\{ \big( \bS \bS^T \big)^j \Big\} \label{eq:chi}
    \end{equation}
    
    \noindent
    Additionally, the variance $v_{h_t}^i$ from \eqref{eq:v_h_WS} evolves as
    
    \noindent
    \begin{equation}
        v_{h_t}^i \overset{a.s.}{=} \Big(\sum_{\tau = 0}^t \gamma_t^{\tau,i}\Big)^{-2} \Big(\Omega_t^i - \sum_{\tau, \tau' = 0}^t \gamma_t^{\tau,i} \gamma_t^{\tau',i} \psi_{\tau,\tau'}\Big)  \label{eq:v_h_WS_CG}
    \end{equation}
    
    \noindent
    with
    
    \noindent
    \begin{equation}
        \Omega_t^i = \sum_{\tau,\tau' = 0}^{t} \sum_{j = 0}^{(t - \tau) i} \sum_{k = 0}^{(t - \tau') i} r_{\tau}^{t,i}[j] r_{\tau'}^{t,i}[k] \big(\psi_{\tau,\tau'} \chi_{j+k+2} \nonumber\\
    \end{equation}
    
    \noindent
    \begin{equation}
        + v_w \chi_{j+k+1} \big) \label{eq:Omega}
    \end{equation}
    

\end{theorem}

In practice, to implement \eqref{eq:chi} without the direct access to the eigenvalue spectrum of $\bA \bA^T$, one can use the Monte Carlo technique \cite{MC_matrix_moments_1, MC_matrix_moments_2, MC-divergence} to estimate $\chi_j$ by referring to $\bA$ only. Additionally, one can estimate the scalar $\psi_{\tau,\tau'}$ based on the following asymptotic identity \cite{OurFullPaper}

\noindent
\begin{equation}
    \lim_{N \rightarrow \infty} \inv{N} \bz_{\tau}^T \bz_{\tau'} - \delta v_w \as \psi_{\tau,\tau'} \label{eq:psi_LSL_estimator}
\end{equation}

While \eqref{eq:gamma_WS_CG} and \eqref{eq:v_h_WS_CG} provide the asymptotic identities for $\gamma_t^{\tau, i}$ and $v_{h_t}^i$ and can be used to study theoretical properties of WS algorithms for approximating \eqref{eq:LMMSE}, unfortunately in practice constructing estimators based on these results might lead to stability issues. For this reason, we suggest to use a different pair of estimates of those scalars. For $v_{h_t}$, we can utilize the property \eqref{eq:h_q_independence} of any OAMP/VAMP-based algorithms to obtain

\noindent
\begin{align}
    \lim_{N \rightarrow \infty} \inv{N} &||\br_t - \bs_t||^2 - v_{q_t} = \lim_{N \rightarrow \infty} \inv{N} ||\bh_t - \bq_t||^2 - v_{q_t} \nonumber\\
    &\overset{a.s.}{=} v_{h_t} + v_{q_t} - v_{q_t} = v_{h_t} \label{eq:robust_v_h_estimator}
\end{align}

\noindent
Thus, we suggest to use $\hat{v}_{h_t}^i = \inv{N} ||\br_t - \bs_t||^2 - \hat{v}_{q_t}$ as the estimator of $v_{h_t}^i$. To work out a more robust way of estimating $\gamma_t^{\tau, i}$, first, we can use \eqref{eq:y_measurements} and $\bs_t = \bx + \bq_t$ to show that $\bz_t = \bw - \bA \bq_t$. Then, by defining $\bZ_t = (\bz_0, ...,\bz_t)$, $(\bPsi_t)_{\tau,\tau'} = \psi_{\tau,\tau'}$ and $(\bgamma_t^i)_{\tau} = \gamma_t^{\tau,i}$, we can use \eqref{eq:q_A_mu_WS_CG} to obtain

\noindent
\begin{equation}
    \bgamma_t^i \as \lim_{N \rightarrow \infty} \bPsi_t^{-1} \big( \inv{N} \bZ_t^T \bmu_t^i - \inv{N} \bw^T \bmu_t^i \bone_t \big), \label{eq:Gamma_t_i_WS_CG}
\end{equation}

\noindent
where $\bone_t$ is a column vector of ones of dimension $t$ and the invertibility of $\bPsi_t$ was confirmed in \cite{CG_EP}. Since the product $\bZ_t^T \bmu_t^i$ is available, we can leverage the iterative structure of \eqref{eq:general_iterative_scheme} to work out the last unknown component $\inv{N} \bw^T \bmu_t^i$. For this, we define the following recursion

\noindent
\begin{gather}
    \nu_t^{i+1} = \nu_t^{i} + a_t^{i} \eta_t^{i} \label{eq:nu} \\
    \bsigma_t^{i+1} = \bPsi_t^{-1} \Big( \inv{N} \bZ_t^T \bmu_t^{i+1} - \nu_t^{i+1} \bone_t \Big) \label{eq:Gamma_t_i_WS_CG}\\
    \eta_t^{i+1} = v_w \Big(\delta - \frac{\nu_t^{i+1}}{v_{q_t}} - \big|\big|\bsigma_t^{i+1}\big|\big|_1 \Big) + b_t^{i} \eta_t^i \label{eq:eta}
\end{gather}

\noindent
with the initializations $\nu_0^0 = 0$, $\eta_0^0 = \delta v_w$ for $t=0$, and $\nu_t^0 = \nu_{t-1}^i$, $\eta_t^0 = v_w \big( \delta - \frac{\nu_{t-1}^i}{v_{q_t}} - \big|\big|\bsigma_{t-1}^i\big|\big|_1 \big) + b_{t-1}^{i-1} \eta_{t-1}^{i-1}$ for $t>0$. 

The following theorem establishes an alternative asymptotic identity for $\bgamma_t^i$.

\begin{theorem} \label{theorem:alternative_gamma_t_tau_estimator}
    
    \textit{\cite{OurFullPaper}}: Let $\bmu_t^i$ be the output of iterations \eqref{eq:general_iterative_scheme} with the initializations \eqref{eq:general_initialization}. Let $\bsigma_t^{i+1}$ be computed from \eqref{eq:nu}-\eqref{eq:eta} where $a_t^i$ and $b_t^i$ are finite scalars. Then, under Assumptions 1-2, we have that
    
    \noindent
    \begin{equation}
        \limN \bgamma_t^i \overset{a.s.}{=} \limN \bsigma_t^{i} \label{eq:iterative_gamma_WS_CG}
    \end{equation}

\end{theorem}

Thus, in practice we suggest to use the empirical version of \eqref{eq:iterative_gamma_WS_CG} where $v_w$, $v_{q_t}$ and $\bPsi_t$ are substituted by their estimated versions. Importantly, the resulting estimator does not explicitly use the moments $\chi_j$, which leads to higher robustness of the resulting MP algorithm utilizing \eqref{eq:general_iterative_scheme} with \eqref{eq:general_initialization} as will be shown next.

\section{Numerical Experiments}

In this section we numerically compare the performance of the LM MP algorithms utilizing WS-CG and WS-GD for approximating \eqref{eq:LMMSE} against VAMP \cite{VAMP}, regular CG-VAMP \cite{CG_EP, OurFullPaper} and the optimized version of MAMP \cite{MAMP_full}. For this, we consider a problem of recovering a natural image 'man' of dimension $1024$ by $1024$ measured by a Fast ill-conditioned Johnson-Lindenstrauss Transform (FIJL). Here, the FIJL operator $\bA = \bS \bP \bH \bD$ is composed of \cite{NS-VAMP}: a diagonal matrix $\bD$ with random values either $-1$ or $1$ with equal probability; a Discreet Cosine Transform matrix $\bH$; a random permutation matrix $\bP$ and an $M$ by $N$ matrix $\bS$ of zeros except for the main diagonal, where the singular values follow the geometric progression leading to the desired condition number. Note that 

\noindent
\begin{equation}
    \bA \bA^T = \bS \bP \bH \bD \bD^T \bH^T \bP^T \bS^T = \bS \bS^T
\end{equation}

\noindent
is diagonal, which implies we can use oracle information about $\bS \bS^T$ to compute the LMMSE estimator \eqref{eq:LMMSE} and implement VAMP as a benchmark. However, here we emphasize that apart from VAMP, all other algorithms do not utilize this knowledge and treat $\bA$ as a general matrix. 

In the following experiments, we set the condition number of $\bA$ to be $\kappa(\bA) = 10^3$, subsampling factor $\delta = 0.05$ and the measurement noise variance $v_w$ to achieve SNR $\frac{||\bx||^2}{||\bw||^2}$ of $40dB$. All the algorithms use the same update of $\bs_{t+1}$ from \eqref{eq:s_t}, where we implement the BM3D denoiser \cite{BM3D}. Additionally, for all the algorithms we estimate the divergence $\alpha_t$ from \eqref{eq:gamma_alpha_general} using the method proposed in \cite{Our_divergence_paper}, apart from MAMP case for which we use the Black-Box Monte Carlo (BB-MC) method \cite{MC-divergence} to improve the convergence properties. The method from \cite{Our_divergence_paper} allows executing the denoiser $\bg(\cdot)$ half as many times as BB-MC, which effectively leads to an acceleration of the algorithm by almost factor of two. Furthermore, for those algorithms that require the access to the spectral moments $\chi_j$ from \eqref{eq:chi}, we generate their estimates $\Hat{\chi}_j$ using the MC method from \cite{MC_matrix_moments_1, MC_matrix_moments_2, MC-divergence} with $10^3$ MC trials.

In the first experiment we compare the stability of two pairs of WS-CG-VAMP and WS-GD-VAMP utilizing WS-CG and WS-GD from section \ref{sec:LM_LMMSE_approximators} with $i=1$ and $i=10$. The first pair uses \eqref{eq:gamma_WS_CG} and \eqref{eq:v_h_WS_CG} from Theorem \ref{theorem:WS_CG_correction_and_SE} with the oracle values of $\chi_j$, $v_w$ and $\psi_{\tau,\tau'}$ substituted by their estimated versions $\Hat{\chi}_j$, $\hvw$ and $\Hat{\psi}_{\tau,\tau'}$. We refer to this pair as \textit{WS-CG-VAMP A} and \textit{WS-GD-VAMP A} respectively. The second pair, which will be referred as \textit{WS-CG-VAMP B} and \textit{WS-GD-VAMP B}, utilizes the alternative estimators based on \eqref{eq:robust_v_h_estimator} and \eqref{eq:iterative_gamma_WS_CG}. For both WS-GD-VAMP A and WS-GD-VAMP B with $i=1$, we employ the same damping strategy of length $l=3$ for $\bs_t$ as has been proposed in MAMP \cite{MAMP_full}. This was necessary since without damping the algorithms became unstable. Here, and further in MAMP, the damping of length $l$ computes the linear combination of the last $l$ vectors $\big\{ \bs_{\tau} \big\}_{\tau = t-l+1}^t$ where the weights are optimized to minimize an estimate of the variance $v_{q_t}$ of the error $\bq_{t} = \bs_{t} - \bx$. To evaluate the performance, we compute the oracle Normalized Mean Squared Error (NMSE) $\frac{||\bg(\br_t) - \bx||^2}{||\bx||^2}$ of the algorithms. The NMSE of the two pairs of algorithms averaged over $10$ realizations is shown on Figure \ref{fig:WS_A_B_comparison}. As seen from the plot, the \textit{A} version of all the algorithms eventually becomes unstable, which we relate to high sensitivity of the estimators based on \eqref{eq:gamma_WS_CG} and \eqref{eq:v_h_WS_CG} to the estimation error of the parameters $v_w$ and $\psi_{\tau,\tau'}$. A more detailed discussion of this nuance can be found in the full work \cite{OurFullPaper}. On the other hand, the \textit{B} version of the algorithms demonstrates stable and consistent progression for a larger $t$. Additionally, we see that using WS-CG leads to a faster dynamics compared to using WS-GD, so in the following experiments we stick to WS-CG-VAMP \textit{B}.

\begin{figure}
\centering
\includegraphics[width=0.45\textwidth]{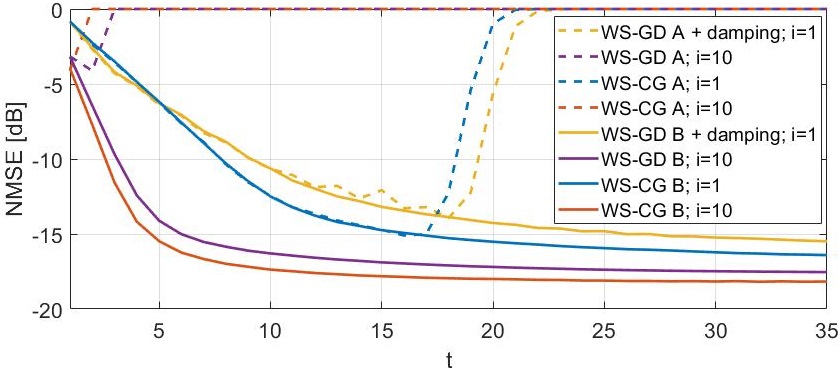}
\caption{NMSE over iteration number $t$ for the LM MP algorithms \eqref{eq:r_t}-\eqref{eq:s_t} with different WS approximators.}
\label{fig:WS_A_B_comparison}
\end{figure}

In the next experiment, we compare WS-CG-VAMP against VAMP, CG-VAMP and the optimized version of MAMP that uses the damping of length $l=3$ to ensure stability of the algorithm. Additionally, we implement two MAMP algorithms: one that uses the ground truth moments $\chi_j$ computed through \eqref{eq:chi} assuming the oracle access to $\bS \bS^T$ and one with the estimated moments $\Hat{\chi}_j$. The NMSE of the algorithms averaged over $10$ iterations is shown on Figure \ref{fig:different_algorithms}. As demonstrated on the plot, WS-CG-VAMP with $i=10$ gets close to the fixed point of VAMP within $20$ outer-loop iterations, while the same algorithm but with $i=1$ has a consistent, but slower convergence and will take hundreds of iterations to achieve the same accuracy level. Similar dynamics are exhibited for MAMP with access to the oracle moments $\chi_j$, while the practical version of the same algorithm converges to an inferior fixed point. Although doing fewer inner-loop iterations leads to a faster computation of $\bmu_t^i$, after some $t$ the summation of vectors $\bs_{\tau}$ in \eqref{eq:r_t} becomes as expensive as doing multiple inner-loop iterations of WS-CG or WS-GD when $\bA$ has a fast implementation. Additionally, in some applications like imaging, executing the denoiser has a substantial computational portion of one outer-loop iteration. As a result, one might want avoid doing a large number of outer-loop iterations and balance the workload between the linear and the denoising steps. WS-CG-VAMP provides such an option.

\begin{figure}
\centering
\includegraphics[width=0.45\textwidth]{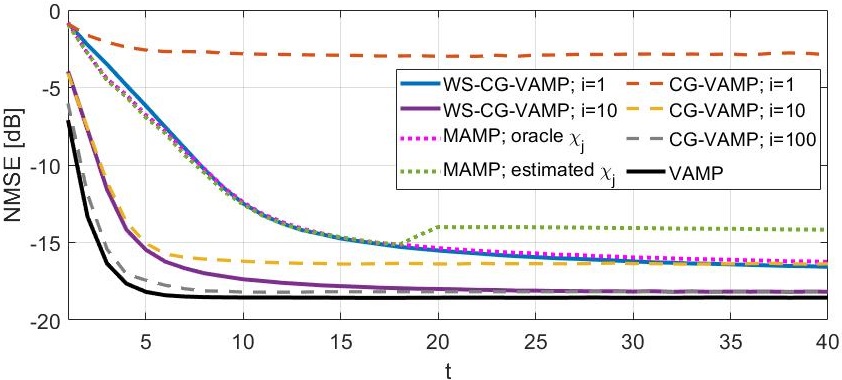}
\caption{NMSE over iteration number $t$ for different Message Passing algorithms.}
\label{fig:different_algorithms}
\end{figure}

\section{Conclusions}

In this work we proposed an iterative WS framework for approximating the computationally intractable LMMSE estimator to upscale the VAMP/OAMP algorithm. When an LM MP algorithm uses an instance of this framework, such as WS-GD or WS-CG, with an arbitrary number of inner-loop iterations, the fixed point of such an MP algorithm is identical to the one of VAMP. Additionally, we provided the asymptotic expressions for the Onsager correction and the SE of the WS framework, and proposed fast and robust methods for estimating those parameters without any knowledge about the singular spectrum of $\bA$. A special case of an algorithm employing the WS framework, WS-CG-VAMP, demonstrates fast and stable convergence even when the measurement operator is highly undersampled and ill-conditioned. Lastly, we show that the approximation method used in MAMP \cite{MAMP_full, MAMP_short} can be interpreted as a particular instance of the WS-GD algorithm with $1$ inner-loop iteration.

\renewcommand*{\bibfont}{\footnotesize}
\printbibliography

\end{document}